\documentclass[orivec]{llncs}

\pdfoutput=1

\usepackage[utf8]{inputenc}
\usepackage[T1]{fontenc}

\usepackage{amsthm}
\usepackage[hidelinks]{hyperref}
\usepackage{url}
\usepackage{amsfonts}
\usepackage{nicefrac}
\usepackage{graphicx}
\usepackage{amssymb}
\usepackage{amsmath}
\usepackage{caption}
\usepackage{subcaption}
\usepackage{doi}
\usepackage{tikz}
\usepackage{mathtools}
\usepackage{algorithm}
\usepackage{algpseudocode}
\usepackage{url}

\usepackage{geometry}
\newcommand{\eps}{\ensuremath{\varepsilon}}
\newcommand{\nat}{\mathbb N}

\newcommand{\sig}{\ensuremath{\Sigma}}
\newcommand{\lang}{\mathcal{L}}

\newcommand{\fresh}{\ensuremath{\textit{fresh}}}

\newcommand{\degloc}{\textit{dloc}}

\makeatletter
\long\def\ifnodedefined#1#2#3{%
    \@ifundefined{pgf@sh@ns@#1}{#3}{#2}%
}
\makeatother

\usepackage{amsmath}
\usepackage{amssymb}
\usepackage{slashed}
\usepackage{wrapfig}
\usepackage{MnSymbol} 
\usepackage{stmaryrd}
\usepackage{mathtools}
\usepackage{moreverb}
\usepackage{enumerate}
\usepackage[commandnameprefix=always]{changes}

\usepackage{tikz}
\usepackage{tikzpagenodes}
\usetikzlibrary{positioning,automata,shapes,calc,patterns,patterns.meta,decorations.pathreplacing,arrows,arrows.meta}

\tikzset{>=latex'}

\newcommand{\change}[1]{\textcolor{black}{#1}}
\renewcommand{\emptyset}{\varnothing}

\pgfdeclarelayer{background}
\pgfdeclarelayer{doublebackground}
\pgfdeclarelayer{triplebackground}
\pgfsetlayers{triplebackground,doublebackground,background,main}

\makeatletter
\newcommand\footnoteref[1]{\protected@xdef\@thefnmark{\ref{#1}}\@footnotemark}
\makeatother

\newcommand{\tokenbox}[1]{%
  \tikz[baseline] {%
    \node[outer sep=0,inner sep=0,anchor=base west] (text) {\small \hskip.5pt #1\hskip1.3pt\vphantom{$\hat{I}$}};%
    \begin{pgfonlayer}{background}%
      \fill[white] (-.1em,-.2em) rectangle (text.north east);%
      \draw[overlay,very thick,red!20!white] (-.1em,-.2em) rectangle (text.north east);%
    \end{pgfonlayer}
  }%
}

\newcommand{\greenback}[1]{%
  \tikz[baseline] {%
    \node[outer sep=0,inner sep=0,anchor=base west] (text) {\small \hskip.5pt #1\hskip1.3pt\vphantom{$\hat{I}$}};%
    \begin{pgfonlayer}{background}%
      \fill[green,opacity=0.3] (-.1em,-.2em) rectangle (text.north east);%
    \end{pgfonlayer}
  }%
}

\newcommand{\visspace}{%
    \tikz[baseline] {%
        \draw (0,0.2ex) -- (0,0) -- (0.6ex,0) -- (0.6ex,0.2ex);
        \path (0,0) to (0.7ex,0);
    }%
}

\theoremstyle{remark}
\newtheorem{observation}{Observation}

\usetikzlibrary{automata, positioning, arrows, shapes}

\newcommand{\tok}{\ensuremath{\mathord{\,\raisebox{0ex}{\ensuremath{\wr}}\,}}}

\newcommand{\xto}[1]{\mathrel{\tikz[every node/.style={inner sep=0}] {%
	\node[overlay] (l) {$\scriptstyle #1$}; %
	\draw ($(l.south west)+(-3pt,-1pt)$) edge[->] ($(l.south east)+(5pt,-1pt)$);
	\node at ($(l.south)+(0,-2pt)$) {};
}}}

\title{Constructing a BPE Tokenization DFA\protect\footnote{This is a revised and extended version of~\cite{ciaa-tokenizing-automata}. Prefer citing that version whenever the extended content is unnecessary.}}

\author{Martin Berglund\inst{1} \and Willeke Martens\inst{1} \and Brink van der Merwe\inst{2,3}}

\institute{Department of Computing Science, Umeå University,\\\email{\href{mailto:mbe@cs.umu.se}{mbe@cs.umu.se}}, \email{\href{mailto:willeke.martens@cs.umu.se}{willeke.martens@cs.umu.se}} \and 
Department of Computer Science, Stellenbosch University, \email{\href{mailto:abvdm@cs.sun.ac.za}{abvdm@cs.sun.ac.za}} \and National Institute for Theoretical and Computational Sciences, South
Africa}

\begin{document}

\maketitle

\tikzset{
	autbase/.style={
		->,
		>=latex',
		node distance=1.75cm,
		every state/.style={minimum size=1.5em,font=\footnotesize,inner sep=1pt},
		initial text={},}
}

\tikzset{
	autwoarrows/.style={
		>=latex',
		node distance=1.75cm,
		every state/.style={minimum size=1.5em,font=\footnotesize,inner sep=1pt},
		initial text={},}
}

\begin{abstract}\label{sec:abstract}Many natural language processing systems operate over \emph{tokenizations} of text to address the open-vocabulary problem. In this paper, we give and analyze an algorithm for the efficient construction of deterministic finite automata (DFA) designed to operate directly on tokenizations produced by the popular byte pair encoding (BPE) technique. This makes it possible to apply many existing techniques and algorithms to the tokenized case, such as pattern matching, equivalence checking of tokenization dictionaries, and composing tokenized languages in various ways. The construction preserves some key properties of the automaton, and we use this to establish asymptotic bounds on the state complexity of the automata that result. Finally, we demonstrate how to construct an input-deterministic (subsequential) string-to-string transducer which precisely describes the relationship between strings and their correct tokenizations.
\end{abstract}

\section{Introduction}\label{sec:introduction}
Subword tokenization, which decomposes a string into smaller textual units, is a key strategy for handling the open-vocabulary problem in modern large language models. Notably, byte pair encoding (BPE) tokenization~\cite{sennrich}, grounded in data compression, is used in the context of prominent models like OpenAI's GPT series~\cite{gpt2} (underpinning ChatGPT~\cite{chatgpt}) and Meta's LLaMA series. The algorithm successively merges adjacent tokens according to a dictionary of rules, typically built based on token pair frequencies in a training corpus.

For example, the GPT-2 tokenization of this string \tokenbox{CIA}\tokenbox{A}\tokenbox{\visspace 2024}\tokenbox{\visspace in}\tokenbox{\visspace Ak}\tokenbox{ita}\tokenbox{,} \tokenbox{\visspace Japan} is indicated by the boxes. Common substrings like `2024', `in' and `Japan' are treated as single tokens, while `CIAA' and `Akita' are split into more common substrings, such as the well-known acronym `CIA'.
The GPT-2 dictionary can also generate the tokens \tokenbox{\visspace CI}, \tokenbox{AA}, \tokenbox{\visspace Aki}, and \tokenbox{ta}, which potentially form alternative tokenizations like \tokenbox{\visspace CI}\tokenbox{AA} for `CIAA' and \tokenbox{\visspace Aki}\tokenbox{ta} for `Akita'. However, merge rules in token dictionaries have distinct priorities, ensuring that the BPE algorithm, given in Algorithm~\ref{alg:hf}, maps every string to precisely one tokenization. In this case, the correct tokenization, is the one introduced earlier:\tokenbox{CIA}\tokenbox{A} and \tokenbox{\visspace Ak}\tokenbox{ita}.

While every string is mapped to a unique tokenization, the tokenization of a particular substring depends on its surrounding context. For example, \tokenbox{\visspace kit} is a single token when appearing on its own, but in \tokenbox{\visspace Ak}\tokenbox{ita}, the same substring is split across two tokens. This complicates traditional string processing tasks, such as pattern matching, typically concerned with strings
over a fixed alphabet of symbols, where every string is denoted by precisely
one sequence of symbols (i.e. it is based on the free monoid). As pointed out,
distinct sequences of tokens can contain the same substring. Searching for a pattern in the underlying string of a tokenization requires accounting for multiple possible tokenizations. %

Such automata are of great practical importance. Token sequences are already common in communication with language services, and if such services keep getting more popular, more and more text may be encoded in such a way. Specific use cases may include pattern matching for filtering (e.g., on the network or for safety systems), validating the correctness of tokenizations (incorrect tokenizations can severely confuse models, similarly to occurrences of extremely rare tokens~\cite{glitch-tokens}), and to \emph{perform} tokenizations or rewrite tokenized text directly. However, for practical applications it is important that the automata are not excessively large. In Section~\ref{sec:complexitybounds} we establish that the algorithm producing the automata preserves a certain measure of \emph{locality}, and that this measure also bounds the number of states the procedure can create, ensuring that the resulting automata are of modest size.

Often the representation of the correct tokenizations is not the primary interest, but rather the \emph{relation} between strings and their tokenization. This too can be efficiently represented, i.e.\ in Section~\ref{sec:transducer} we demonstrate how to obtain a string-to-string (or rather, string-to-\emph{tokenization}) transducer which represents this relation. It is further demonstrated that this transducer can be made subsequential (i.e.\ deterministic in its input), so it can be used to efficiently tokenize text.

\section{Notation and Basic Definitions}
An \emph{alphabet} $\Sigma$ is a finite set of symbols. As we also consider alphabets of strings (or `tokens'), we consistently refer to a normal alphabet consisting of `indivisible' symbols as a \emph{base alphabet.} Let $\Sigma^*$ denote the set of all strings (including the empty string $\eps$) over the base alphabet $\Sigma$ and let $\Sigma^+=\Sigma^*\setminus \{\eps\}$.

A \emph{token alphabet} $\Gamma$ over the base alphabet $\Sigma$ is a finite subset $\Gamma\subset \Sigma^+$ where $\Sigma \subseteq \Gamma$ and for all $w\in \Gamma$ with $|w|>1$ there are some $u,v\in \Gamma$ such that $w=uv$. We refer to the elements of such $\Gamma$ as \emph{tokens.} A sequence of tokens $u_1,\ldots,u_n \in \Gamma$ is denoted $u_1 \tok \cdots \tok u_n$ and called a \emph{tokenization}. The set of all sequences of tokens over $\Gamma$ is denoted $\Gamma^{\tok}$. For any arbitrary token alphabet $\Gamma$ and corresponding base alphabet $\Sigma$, we define $\pi : \Gamma^{\tok} \to \Sigma^*$ as $\pi(u_1\tok \cdots \tok u_n)=u_1 \cdots u_n$, i.e.\ the concatenation of the tokens. 
For a tokenization $\tau$, let $|\tau|$ denote the number of tokens in $\tau$, i.e.\ $|u_1\tok \cdots \tok u_n|=n$ for $u_1,\ldots,u_n\in \Gamma$. 

Observe that while both tokenizations and strings are sequences, they are different \emph{kinds} of sequences even when they contain the same elements. E.g., the tokenization $\alpha_1 \tok \cdots \tok \alpha_n \in \Gamma^{\tok}$ is distinct from the string $\alpha_1 \cdots \alpha_n$ despite $\alpha_1,\ldots,\alpha_n\in \Sigma$. That is, they are both obtained by using the same sequence of elements from $\Sigma$, but are combined with the concatenation and tokenization operators, respectively. The tokenization $\alpha_1 \tok \cdots \tok \alpha_n \in \Gamma^{\tok}$ with $\alpha_i\in\Sigma\subseteq\Gamma$, is called the \emph{base tokenization} of the string $\alpha_1\cdots\alpha_n$. 

To keep our exposition clear, we adopt some conventions. We let $\Sigma$ denote some base alphabet, and $\Gamma$ a token alphabet over $\Sigma$. When giving examples, we use $\Sigma=\{\alpha,\beta,\gamma,\ldots\}$. Furthermore, we let $\alpha, \beta, \gamma$ be variables denoting symbols (so elements from the base alphabet), let $u,v,w$ be strings or
tokens,
and $\tau,\phi$ tokenizations. In each case, we also reserve all sub-/super-scripted variants of these symbols for those same purposes. As such we have, for example, $\alpha_3 \in \Sigma$ and $\hat{\tau} \in \Gamma^{\tok}$. When writing $\phi=u\tok \tau \tok v$, we have that $\phi$ is a tokenization where the first token is $u$, the last token is $v$, and the intervening tokens form the tokenization $\tau$, so $|\phi|=|\tau|+2$.

\begin{definition}
    \label{defn:dfa}%
    A token deterministic finite automaton (DFA) is a quintuple $A=(Q,\Gamma,q_0,\delta,F)$ where:
    
    \begin{itemize}
        \item $Q$ is the finite set of \emph{states,}
        \item $\Gamma$ is the \emph{token alphabet,}
        \item $q_0\in Q$ is the \emph{initial state,}
        \item $\delta : Q\times \Gamma \to Q$ is the \emph{transition function,} which may be partial,
        \item $F\subseteq Q$ is the set of \emph{final states.}
    \end{itemize}
    A \emph{run} of $A$ is a sequence $q_1,u_1,q_2,u_2,q_3,\ldots,u_{n-1},q_{n}$ for some $n\ge 0$, $q_1,\ldots,q_n \in Q$, $u_1,\ldots,u_{n-1}\in \Gamma$, such that $\delta(q_{i},u_i)=q_{i+1}$ for all $1 \le i < n$. We denote such a run as $q_1 \xto{u_1} q_2 \xto{u_2} \cdots \xto{u_{n-1}} q_n$, or with the shorthand $q_1 \xto{u_1\tok\cdots\tok u_{n-1}}q_n$ if the intermediate states are not of interest. We say that such a run \emph{reads the string $\pi(u_1\tok \cdots \tok u_{n-1})$} or \emph{tokenization $u_1\tok\cdots\tok u_n$.} %
    
    The run is \emph{accepting} iff $q_1=q_0$ and $q_n\in F$. The \emph{language accepted by $A$,} denoted $\lang(A)\subseteq \Gamma^{\tok}$ is the set of all tokenizations $\tau$ for which there exists an accepting run $q_0 \xto{\tau} q_n$.

\end{definition}

Token DFAs give rise to some additional questions about determinism. Consider the automaton for the (finite) language $\{a\tok bc, a\tok b\tok c \}$, as shown in Fig.~\ref{fig:token-aut-non-base-det}.
\begin{figure}
    \centering
    \begin{tikzpicture}[autbase,node distance=1.25cm]
        \node[state, initial] (q0) {$q_0$};
        \node[state, right of=q0] (q1) {$q_1$};
        \node[state, right of=q1] (q2) {$q_2$};
        \node[state, accepting, right of=q2] (q3) {$q_3$};
        \node[state, accepting, below of=q1,yshift=.5em] (q4) {$q_4$};
        \draw (q0) edge[above] node{$a$} (q1)
        (q1) edge[above] node{$b$} (q2)
        (q2) edge[above] node {$c$} (q3)
        (q1) edge[left,pos=0.4] node{$bc$} (q4);
    \end{tikzpicture}
    \caption{A token DFA for the language $\{a\tok b \tok c, a\tok bc\}$ which indicates the existence of two distinct tokenizations of the string $abc$.}
    \label{fig:token-aut-non-base-det}
\end{figure}

This automaton is deterministic in the classical sense, but there are two distinct accepting runs which read $abc$ (i.e.\ $\pi(a\tok b\tok c)=\pi(a\tok bc)$). As we aim to use token DFAs to represent only correct BPE tokenizations of strings, where (as we will see later) this situation cannot occur, we introduce a stronger condition which precludes such situations.

\begin{definition}\label{def:pref_invariant}
A token DFA $A=(Q,\Gamma,q_0,\delta,F)$ is \emph{context-invariant} if for runs $q_1\xto{\varphi}q_n$ and $q_1'\xto{\varphi'}q_m'$ in $A$  
 for which $\pi(\varphi)=\pi(\varphi')$ (with $q_1$ or $q_1'$ not necessarily being equal to $q_0$), we have $\varphi=\varphi'$.\end{definition}

\begin{remark}\label{rem:pref_invariant}
    Context-invariance ensures that every string has a unique tokenization, e.g., the token DFA in Figure~\ref{fig:token-aut-non-base-det} is not context-invariant, as its two runs $q_0 \xto{a \tok b \tok c} q_3$ and $q_0 \xto{a\tok bc} q_4$ violate Definition~\ref{def:pref_invariant}. Context-invariance is a stronger property than simply requiring the uniqueness of tokenizations; consider the automaton in Figure~\ref{fig:unique-but-not-ci}.
    \begin{figure}[htb]
        \centering
        \begin{tikzpicture}[autbase,node distance=1.25cm]
            \node[state, initial] (q0) {$q_0$};
            \node[state, right of=q0,accepting] (q1) {$q_1$};
            
            \draw (q0) edge[loop above] node{$a$} (q0);
            \draw (q0) edge[above] node {$aa$} (q1);
        \end{tikzpicture}
        \caption[A token DFA with unique tokenizations that is not context-invariant]{A token DFA which has unique tokenizations $\{aa, a\tok aa, a\tok a\tok a\tok aa, \ldots\}$ but is not context-invariant, with runs such as $q_0 \xto{a\tok aa} q_1$ and $q_0 \xto{a \tok a \tok a} q_0$ violating the property.}
        \label{fig:unique-but-not-ci}
    \end{figure}
    Uniqueness is the special case of Definition~\ref{def:pref_invariant} where we only consider runs where $q_1=q_1'$ is the initial state, and both runs are accepting. However, as we will see, BPE tokenization does fulfill context-invariance, so we demonstrate this stronger property for the automata our algorithm produces.

  \end{remark}

Next, we define BPE tokenization as defined by~\cite{sennrich}, using the formal structure introduced in~\cite{berg2023}.
\begin{definition}
    A \emph{byte pair dictionary $D$} over $\Sigma$ is a sequence of pairs of tokens over $\Sigma$, denoted $D=[u_1\tok v_1,\ldots,u_n \tok v_n]$. We call each $u_i \tok v_i$ a \emph{rule,} and say that $u_i\tok v_i$ has \emph{higher priority} than $u_j \tok v_j$ if $i<j$.

    A dictionary is \emph{proper} if for each $j$ with $|u_j|>1$ there exists some $i<j$ such that $u_j=u_iv_i$, and, symmetrically, for each $j$ with $|v_j|>1$ there exists some $i<j$ such that $v_j=u_iv_i$.
    
    Observe that when $D$ is proper this makes $\Sigma\cup\{u_1v_1,\ldots,u_nv_n\}$ a valid token alphabet, which we call the \emph{token alphabet of $D$.}
\end{definition}

Unless otherwise stated, we assume that all dictionaries are proper. We adopt the convention that $D$ and all its sub-/superscripted variants, always refer to some (proper) dictionary. In~\cite{berg2023}, both SentencePiece~\cite{sentencepiece} and HuggingFace~\cite{huggingface-gpt-2-py} BPE tokenization semantics are described, but the semantics are demonstrated to coincide for proper dictionaries. Some of the later arguments are clearer for HuggingFace semantics, so we choose to use those in Algorithm \ref{alg:hf}. SentencePiece semantics are arguably \emph{more natural} however (in~\cite{berg2023} we term them `correct'), as they simply amount to picking the highest-priority applicable rule after each time two tokens are merged. That is, to obtain the SentencePiece BPE tokenization algorithm, lines 3 and 8 in Algorithm \ref{alg:hf} are removed, and line 5 is replaced by the following: ``Let $u_i\tok v_i$ be the highest priority rule in $D$ for which a decomposition, as specified in line 4, exists. Then, among the remaining decompositions, pick the unique one which minimizes $|\phi|$.''

\begin{algorithm}
    \caption{HuggingFace BPE tokenization}
    \label{alg:hf}%
    \begin{algorithmic}[1]
        \State \textbf{Input:} a string $\alpha_1 \cdots \alpha_n \in \Sigma^*$, a proper dictionary $D=[u_1\tok v_1,\ldots,u_m \tok v_m]$
        \State initialize $\tau=\alpha_1 \tok \cdots \tok \alpha_n$
        \For {each rule $u_i\tok v_i$ in $D$ in priority order: \label{alg:hf:loop}}
            \While {there are $\phi$ and $\phi'$ such that $\phi \tok u_i \tok v_i \tok \phi' = \tau$}
                \State choose such $\phi$ and $\phi'$ minimizing $|\phi|$
                \State update $\tau=\phi \tok u_iv_i \tok \phi'$
            \EndWhile
        \EndFor
        \State \textbf{Output:} the HuggingFace tokenization $\tau$.
    \end{algorithmic}
\end{algorithm} 

We denote the BPE tokenization of a string $w$ according to the dictionary $D$ as $\mathbb{T}^D(w)$. Note that $\mathbb{T}^\varnothing(\alpha_1 \cdots \alpha_n)$ yields $\alpha_1\tok \cdots \tok \alpha_n$, the so-called \emph{base tokenization} of $\alpha_1\cdots \alpha_n$. The definition of $\mathbb{T}^D$ is naturally extended to sets so that $\mathbb{T}^D(L)=\{\mathbb{T}^D(w) \mid w \in L\}$ for a language $L\subseteq \sig^*$. In particular, $\mathbb{T}^D(\sig^+)$ is the set of \emph{all} BPE tokenizations over $\Sigma^+$ according to dictionary $D$.

\begin{example}
    Let $D=[a\tok a, a\tok b,b\tok c,ab\tok c,bc \tok ab]$ and $w=aaaaacbcabc$. Then Algorithm~\ref{alg:hf} proceeds with the following steps:
    \[
    \begin{array}{c}
        \greenback{a\tok a}\tok a\tok a\tok a\tok c\tok b\tok c\tok a\tok b\tok c \;\Rightarrow\; 
        aa\tok \greenback{a\tok a}\tok a\tok c\tok b\tok c\tok a\tok b\tok c \;\Rightarrow\; \\
        aa\tok aa \tok a \tok c\tok b\tok c\tok \greenback{a\tok b}\tok c  \;\Rightarrow\; 
        aa\tok aa\tok a\tok c\tok \greenback{b\tok c}\tok ab \tok c \;\Rightarrow\; \\
        aa\tok aa\tok a\tok c\tok bc\tok \greenback{ab\tok c} \;\Rightarrow\; aa \tok aa\tok a \tok c\tok bc \tok abc
    \end{array}
    \]
    The first two rule applications in this example, happen in the inner loop of Algorithm~\ref{alg:hf} (in two places $a\tok a$ can be applied). Note also that $w$ has two substrings $bc$, but only one use of $b\tok c$ is possible since the second $b$ is used by the higher-priority rule $a \tok b$. Finally, the rule $bc \tok ab$ would be applicable after the fourth step, but $ab \tok c$ has higher priority and uses the $bc$ token. As such we have $\mathbb{T}^D(w)=aa\tok aa \tok a\tok c \tok bc\tok abc$.
\end{example}

\section{Construction Procedure}\label{sec:constructionprocedure}

We now give a procedure for building a token DFA $A$ for a given regular language $L$ and dictionary $D$ over $\Sigma$ such that $\mathcal{L}(A)=\mathbb{T}^{D}(L)$.

For $D=[u_1\tok v_1,\ldots, u_n\tok v_n]$ and $0\le i \le n$, define the dictionary $D_i$ as the prefix $[u_1\tok v_1,\ldots, u_i\tok v_i]$ so that $D_0=[\ ]$ and $D_n=D$. We inductively construct a sequence of DFAs $A_i$ satisfying $\mathcal{L}(A_i)=\mathbb{T}^{D_i}(L)$, for a regular language $L$. The base token DFA $A_0$ (accepting $\mathbb{T}^{\varnothing}(L)$) is trivially obtained from the string DFA for $L$ by changing all symbols in $\Sigma$ into their corresponding tokens in $\Gamma$. For $1\le i \le n$ we construct $A_i$ by merging the rule $u_i \tok v_i$ into $A_{i-1}$ as described in Algorithm~\ref{alg:merge}. Observe that whenever a new transition is created by Algorithm~\ref{alg:merge}, we update the (partial) transition function of the input token DFA, thus yielding a new token DFA.

\newcommand{\remap}{\ensuremath{\textit{remap}}}
\begin{algorithm}[htb]
    \newcommand{\qadd}{Q_{\textrm{add}}}
    \newcommand{\dadd}{\textit{add}_\delta}
    \newcommand{\drem}{\textit{rem}_\delta}
    \caption{Applying a Merge}\label{alg:merge}
    \begin{algorithmic}[1]
    \State \textbf{Input:} a context-invariant DFA $A=(Q,\Gamma,q_0,\delta,F)$, and a rule $u \tok v$
    \State let $S=\{(s_1,s_2,s_3)\in Q^3 \mid \delta(s_1,u)=s_2, \delta(s_2,v)=s_3\}$
    \State let $S_2=\{s_2 \mid (s_1,s_2,s_3)\in S\}$
    \For{$(s_1,s_2,s_3) \in S$}
    \State add new transition by defining $\delta(s_1,uv)=s_3$ \label{alg:merge:add-uv}
    \EndFor
    \State add $uv$ to $\Gamma$
    \For{$s_2 \in S_2$}
    \State create a fresh state from $s_2$, denote it $\fresh(s_2)$
    \State add $\fresh(s_2)$ to $Q$, and if $s_2\in F$ add $\fresh(s_2)$ to $F$ as well \label{alg:merge:add-fresh}
    \EndFor
    \For{every $s_2 \in S_2$}
        \If {$u\neq v$}\label{alg:merge:bgcp}
            \State add new transition by defining $\delta(\fresh(s_2),\alpha)=\delta(s_2,\alpha)$ for all $\alpha \in \Gamma\setminus \{v\}$\label{alg:merge:copy-to-s2p}
        \Else 
            \State add new transition by defining $\delta(\fresh(s_2),\alpha)=\delta(s_2,\alpha)$ for all ${\alpha \in \Gamma\setminus \{v,uv\}}$\label{alg:merge:remove-uv}
        \EndIf \label{alg:merge:endcp}
    \EndFor
    \For{$q\in Q$ with $\delta(q,u)\in S_2$}
    \State replace the transition by defining $\delta(q,u)=\fresh(\delta(q,u))$ \label{alg:merge:add-u}
    \EndFor
    \State output the resulting DFA
    \end{algorithmic}
\end{algorithm}
Algorithm~\ref{alg:merge} merges the rule $u \tok v$ into $A=(Q,\Gamma,q_0,\delta,F)$ by considering each $(s_1,s_2,s_3)\in Q^3$ where $\delta(s_1,u)=s_2$ and $\delta(s_2,v)=s_3$.
These runs reading $u\tok v$ are replaced with ones reading $uv$ instead, by adding the transition $\delta(s_1,uv)=s_3$ (line \ref{alg:merge:add-uv}) and removing the transition $\delta(s_1,u)=s_2$ (line~\ref{alg:merge:add-u}). However, this might inadvertently eliminate runs not containing the sequence $u\tok v$, i.e.\ runs that read a $u$ to transition to $s_2$, but then take a different transition (or no transition) from $s_2$ than the transition on $v$. This is remedied by adding a new state $\fresh(s_2)$ in line \ref{alg:merge:add-fresh}, which inherits all outgoing transitions from $s_2$ except the ones on $v$ and potentially $uv$ in lines \ref{alg:merge:bgcp}-\ref{alg:merge:endcp}. 
In particular, if $u=v$, a transition on $uv$ from $s_2$ should \emph{not} be copied to $\fresh(s_2)$, as this would create the run $s_1 \xto{u} \fresh(s_2) \xto{uu} q$, which should never be part of a run in the DFA produced by~Algorithm~\ref{alg:merge}.
This can most easily be seen by inspecting Algorithm~\ref{alg:hf}, the only way a token $uu$ can be created is by the rule $u\tok u$ being applied, and it is then applied first on the left, producing $uu \tok u$. Formal correctness proofs follow in Section~\ref{sec:correctnessproof}. 
Example~\ref{ex:build} illustrates both the case where $u=v$ and the transition is \emph{not} copied from $s_2$ to $\fresh(s_2)$ (the step from $A_0$ to $A_1$), and the case where $u\ne v$ and the transition \emph{is} copied (the step from $A_1$ to $A_2$, observe that the token $ba$ occurs on two transitions in $A_2$).

\begin{remark}\label{rem:useless}
Algorithm~\ref{alg:merge} will at times produce useless states, which might simply be trimmed afterwards. This does not happen, for example, in the case where
Algorithm~\ref{alg:merge} is applied to a token DFA with all states being accepting.  Notably, this is the case when Algorithm~\ref{alg:merge} is applied to the (one-state) universal token DFA (i.e.\ $A_0$ in Figure~\ref{fig:constructionprocedureex}); or more generally, on token DFA obtained from repeatedly applying Algorithm~\ref{alg:merge} (zero or more times), starting with the universal token DFA.
\end{remark}

\begin{example}
    \label{ex:build}
    For the byte pair dictionary $D=[a\tok a,b\tok a]$ over $\Sigma=\{a,b\}$, constructing the token DFA $A$ that accepts $\mathbb{T}^D(\Sigma^*)$, requires two iterations. Starting from the base token DFA $A_0$ in Figure~\ref{fig:constructionprocedureex}, Algorithm \ref{alg:merge} is first applied with the merge $a\tok a$, yielding the depicted token DFA $A_1$. Next, Algorithm \ref{alg:merge} is applied to $A_1$ with the merge $b\tok a$, generating the token DFA $A_2$ which accepts $\mathbb{T}^D(\Sigma^*)$.
\begin{figure}[htb!]
     \centering
     \begin{tikzpicture}[autbase]
        \node[state, initial, accepting] (q0) {$q_0$};
        \draw (q0) edge[loop above] node{$a,b$} (q0);
        \node[yshift=-1.4cm] (label1) at (q0) {$A_0$ for $\mathbb{T}^\emptyset(\{a,b\}^*)$};

        \node[state, initial, accepting] (q0) at ($(q0)+(3cm,0)$) {$q_0$};
        \node[state, accepting, right of=q0] (q1) {$q_1$};
        \draw (q0) edge[loop above] node{$aa,b$} (q0)
        (q0) edge[bend left, above] node{$a$} (q1)
        (q1) edge[bend left, below] node (b) {$b$} (q0);
        \node at (label1-|b) {$A_1$ for $\mathbb{T}^{D_1}(\{a,b\}^*)$};
        \begin{scope}[node distance=2.4cm]
            \node[state, initial, accepting] (q0) at ($(q0)+(4cm,0cm)$) {$q_0$};
            \node[state, accepting, right of=q0] (q1) {$q_1$};
            \node[state, accepting, right of=q0, yshift=1.2cm, xshift=-1.2cm] (q2) {$q_2$};
            \node at ($(q0)+(0,4em)$) {};

            \draw (q0) edge[loop below] node{$aa$} (q0)
            (q0) edge[bend right=40, below] node {$a, ba$} (q1)
            (q0) edge[bend right=20, below] node{$b$} (q2)
            (q1) edge[loop below] node{$ba$} (q1)
            (q2) edge[bend left=20, above right] node{$ba$} (q1)
            (q2) edge[loop above] node{$b$} (q2)
            (q2) edge[bend right=20, above left] node {$aa$} (q0)
            (q1) edge[bend left=20, below] node{$b$} (q2);
            \node at (label1-|q2) {$A_2$ for $\mathbb{T}^{D}(\{a,b\}^*)$};
        \end{scope}
    \end{tikzpicture}
    \caption{The three token DFA constructed in Example~\ref{ex:build}, $A_0$ is the initial universal base token DFA, $A_1$ is $A_0$ with the rule $a\tok a$ merged, and $A_2$ is $A_1$ with the rule $b\tok a$ merged.}
    \label{fig:constructionprocedureex}
\end{figure}
\end{example}

\begin{remark}
    It is interesting to note when and how we obtain loops at $\fresh(s_2)$ when applying Algorithm~\ref{alg:merge}. This happens when we have a loop on $u$ (in the case where $u\not = v$). Two transitions on $u$ replace this loop, one going from $s_2$ to $\fresh(s_2)$ and the other being a loop on $\fresh(s_2)$. This happens by first in line 12 adding a transition from $\fresh(s_2)$ to $s_2$ (on $u$), and then in line 19, changing the target of this (newly added) transition to be $\fresh(s_2)$ rather than $s_2$, but also adding a transition from $s_2$ to $\fresh(s_2)$ on $u$ (in line 19). An example of this can be seen in Fig.~\ref{fig:constructionprocedureex} \change{$A_1$ to $A_2$}, with the loop on $b$ at $q_0$ in \change{$A_1$} being replaced by a transition from $q_0$ to $q_2$ and a loop at state $q_2$ (both on $b$) in \change{$A_2$}. Contrast this to with what happens when $u=v$, by considering the loop on $a$ at $q_0$ in \change{$A_0$}, and $\fresh(q_0)$ being $q_1$ in \change{$A_1$}.
\end{remark}

\section{Proof of Correctness}\label{sec:correctnessproof}

This section sets out to demonstrate that the procedure in Section~\ref{sec:constructionprocedure} yields the token DFA for the BPE tokenizations of the given language $L$ and dictionary $D$.

In the remainder of this section, we assume that the input DFA in Algorithm~\ref{alg:merge} is context-invariant. The following examples illustrate why this assumption is necessary.
\begin{example}\label{ex:pref-inv-required}
Applying Algorithm~\ref{alg:merge} to the non-context-invariant token DFA in Fig.~\ref{fig:token-aut-non-base-det} with rule $b\tok c$, results in a \emph{nondeterministic} automaton because state $q_1$ gains a second transition on the token $bc$. (But for instance, Lemma~\ref{lem:pairwisecompalt} remains valid otherwise.)\end{example}
\begin{example}
Now, consider the non-context-invariant token DFA  $A$ in Figure~\ref{fig:unique-but-not-ci}, and the rule $a\tok a$. The resulting automaton $A'$ is again nondeterministic. Moreover, Algorithm~\ref{alg:merge} removed, among others, the run on the string $aaa$ from $A'$, despite $a\tok aa\in\mathcal{L}(A)$.
\end{example}

It turns out that Algorithm~\ref{alg:merge} preserves context-invariance. However, before proving this, we first consider the following useful lemma that gets us most of the way in terms of the correctness of Algorithm~\ref{alg:merge}.

\begin{lemma}\label{lem:pairwisecompalt}
    For any rule $u \tok v$ and context-invariant DFA $A=(Q,\Gamma,q_0,\delta,F)$ applying Algorithm~\ref{alg:merge} produces a DFA $A'=(Q',\Gamma',q_0',\delta',F')$ such that there is a run $q_1' \xto{\tau'} q_m'$ in $A'$ if and only if there is a run $q_1 \xto{\tau} q_n$ in $A$ (with $q_1'$ and $q_1$ any state in $Q'$ and $Q$ respectively) where $\pi(\tau')=\pi(\tau)$. Moreover, either both runs start from an initial state or neither does, and either both ends in an accepting state or neither does. Consequently, the runs are either both accepting or rejecting.
\end{lemma}
\begin{proof}
    The lemma trivially holds when in $A$ we have no run of the form $s_1 \xto{u} s_2 \xto{v} s_3$, since then $S=\emptyset$ in Algorithm~\ref{alg:merge}, producing an unchanged $A'=A$. Observe that the lemma therefore especially holds when $A$ contains a transition on $uv$ due to its presupposed context-invariance.
    
    To relate the runs of the automata, we relate their states. Observe that $Q'$ contains either one or two copies of each state from $Q$. Define $\gamma: Q'\to Q$ to recover the state in $A$ from which a state in $A'$ was created, so $\gamma(q)=q$ for $q\in Q$, and $\gamma(\fresh(q))=q$ for all $\fresh(q)\in Q'\setminus Q$ produced by line~\ref{alg:merge:add-fresh}.

    To find a run in $A$ corresponding to the run $q_1' \xto{u_1\tok\cdots\tok u_{m-1}} q_m'$ in $A'$, note that the run $\gamma(q_1') \xto{u_1\tok\cdots\tok u_{m-1}} \gamma(q_m')$ in $A$ is well-defined unless the original run in $A'$ contains a step of the form $q\xto{uv}q'$ since this, by Algorithm~\ref{alg:merge}, is the only transition not directly inherited from $\gamma(q)$. If such a step of the form $q\xto{uv}q'$ exists, then  $(\gamma(q),s,\gamma(q'))\in S$ for some $s\in Q$, and can be replaced by the two steps $\gamma(q)\xto{u}s\xto{v} \gamma(q')$ instead. Repeating this procedure for all steps of the form $q\xto{uv}q'$ produces a run in $A$. Given that $\gamma(q_0')=q_0$ and that $q'\in F'$ if and only if $\gamma(q')\in F$, it follows trivially that the runs agree on accepting (or not accepting).

    In the other direction, we proceed by induction on the length of runs to demonstrate that for all runs $q_1\xto{\tau}q_n$ in $A$, there exists a run $q_1'\xto{\tau'}q_m'$ such that $q_n=\gamma(q_m')$ and $\pi(\tau)=\pi(\tau')$. The base case is trivial since a zero-length run has $\pi(\tau)=\pi(\tau')=\eps$ and $q_1=\gamma(q_1')$ for some $q_1'\in Q'$. As the inductive step, extend an arbitrary run $r$ of length $n-1$ in $A$ with a transition $q_n\xto{w}q_{n+1}$. This changes the corresponding run of $r$ in $A'$ according to the following three cases depending on how $w$ relates to the rule $u\tok v$ being merged.
    \begin{enumerate}
        \item If $w\ne v$: the corresponding run is extended with a transition on $w$ to $\delta'(q_m',w)$, which respects the required relationship between the two runs given that $\delta(\gamma(q'),w)=\gamma(\delta'(q',w))$ for all $q'\in Q$ when $w\ne v$.
        \item If $w=v$ and $q_m'$ is a fresh state: by construction $\delta'(q_m',v)$ is undefined. However, observe that then the \emph{previous} step in the corresponding run \emph{must} be $\delta'(q_{m-1}',u)=q_m'$ (the only way to reach a fresh state). Consequently, $(\gamma(q_{m-1}'),\gamma(q_m'),q) \in S$ in Algorithm~\ref{alg:merge} for some $q\in Q$. Therefore, replace that previous step with $q_{m-1}' \xto{uv} q$. Note that when $u=v$ fresh states do not get transitions on $uv$ (see line~\ref{alg:merge:remove-uv}), but in such a case $q_{m-1}'$ cannot be fresh, as fresh states never have outgoing transitions on $v$.
        \item If $w=v$ and $q_m'$ is \emph{not} a fresh state: just like in case~1 then $\delta'(q_m',w)$ is defined and can be used.
    \end{enumerate}
    By the same argument as above, straightforward inspection of $\gamma$ establishes that either both runs or neither are accepting.
\end{proof}

Now we have the tools to prove that Algorithm \ref{alg:merge} preserves context-invariance.

\begin{lemma}\label{lem:prefixpreserved} 
    For any rule $u \tok v$ and context-invariant DFA $A$, applying Algorithm~\ref{alg:merge} produces a context-invariant DFA $A'$.
\end{lemma}
\begin{proof}
    By contradiction assume that $A'$ is \emph{not} context-invariant. Then $A'$ has runs on $\tau$ and $\varphi$ so that $\tau\ne\varphi$ but $\pi(\tau)=\pi(\varphi)$.

    Let $\tau=w_1\tok \cdots \tok w_n$ and $\varphi = w_1' \tok \cdots \tok w_m'$ where $i$ is the smallest index such that $w_i\ne w_i'$. The mapping defined in the proof of Lemma~\ref{lem:pairwisecompalt} instructs to replace all occurrences of $uv$ with $u\tok v$ in both $\tau$ and $\varphi$ to yield respective corresponding runs $\tau'$ and $\varphi'$ in $A$ that read $\pi(\tau)$. Since $A$ is context-invariant, it must be that $\tau'=\varphi'$, in turn implying $\{w_i,w_i'\}=\{u,uv\}$.
    Without loss of generality let $w_i=u$ and $w_i'=uv$. Then consider the token $w_{i+1}$, which by the defined mapping must have $v$ as a prefix, so either
    \begin{itemize}
        \item $w_{i+1}=v$, but this implies a run $q\xto{u}q'\xto{v} q''$ in $A'$, impossible by Algorithm~\ref{alg:merge}, or,
        \item $u=v$ and then $w_{i+1}=vv$, but line~\ref{alg:merge:remove-uv} prevents the introduction of a transition on $uv$ in the case where $u=v$ if $u$ has just been read.
    \end{itemize}
    No other alternatives for $w_{i+1}$ are possible since any other token with $v$ as a prefix is unaffected by the mapping (in which case $\tau'\ne\varphi'$).  As such, we reach a contradiction, i.e. no such $\tau$ and $\varphi$ exist.
\end{proof}

Next, we show the relationship between tokenizations of two corresponding accepting runs as shown to exist in Lemma~\ref{lem:pairwisecompalt}.

\begin{lemma}\label{lem:correct-tok}
    Let $A=(Q,\Gamma,q_0,\delta,F)$ be a context-invariant token DFA and $A'=(Q',\Gamma',q_0',\delta',F')$ the resulting token DFA of applying Algorithm~\ref{alg:merge} with rule $u\tok v$. Suppose that for a given accepting run $q_1\xto{\tau}q_n$ in $A$, the corresponding accepting run in $A'$  is $q_1'\xto{\tau'}q_m'$ where $\pi(\tau)=\pi(\tau')=w$. Then, if there exists a proper dictionary $D$ such that $\mathbb{T}^D(w)=\tau$, and the dictionary $D'$ obtained by concatenating the rule $u\tok v$ to $D$ is also proper, then $\mathbb{T}^{D'}(w)=\tau'$.
\end{lemma}

\begin{proof}
        From Lemma~\ref{lem:prefixpreserved}, it follows that $A'$ is context-invariant. Hence, we can use the mapping from the proof of Lemma~\ref{lem:pairwisecompalt}, to deduce that $\tau'$ can be obtained from $\tau$ by only altering the subtokenizations $u\tok v$:
        \begin{itemize}
            \item In case $u\neq v$, each such subtokenization is simply replaced by $uv$.
            \item In case $u=v$, $A'$ cannot read the token $u$ followed by a token $uu$, subsequently $u\tok u$ is replaced by $uu$ starting from the subtokenization's left-most occurrence.
        \end{itemize}  
        Now, let $D$ be a proper dictionary such that $\mathbb{T}^D(w)=\varphi$ and suppose $D'$ obtained by concatenating the rule $u\tok v$ to $D$ also is proper. Since Algorithm~\ref{alg:hf} applies each rule $u_i\tok v_i$ in order of priority to $w$, $\mathbb{T}^{D'}(w)$ can be directly obtained by applying the rule $u\tok v$ to $\tau$. When applying $u\tok v$ to $\tau$, we repeatedly decompose $\tau$ into $\phi\tok u\tok v\tok \phi'$ so that $|\phi|$ is minimal and replace $u\tok v$ with $uv$ until no such decomposition can be found. Clearly, this yields that $\mathbb{T}^{D'}(w)=\tau'$.
\end{proof}

Observe that Lemma~\ref{lem:correct-tok} does not hold without the requirement that the involved dictionaries have to be proper. First of all, if a dictionary $D$ is not proper, recall that the SentencePiece and HuggingFace semantics do not necessarily yield the same BPE tokenization, so $\mathbb{T}^{D}$ is not well-defined. Further, we demonstrate by an example that $\mathcal{L}(A)$, in such cases, does not necessarily correspond to either the SentencePiece or HuggingFace tokenizations.
\begin{example}
    Take the dictionary $D=[ab\tok a, a\tok b]$ (which is not proper) and let $L=\{a,b\}^*$. The generated token DFA $A$ is depicted in Fig. \ref{fig:nonproperdict}. (Merging the rule $ab\tok a$ into the base token DFA just results in the base token DFA as no runs on $ab\tok a$ are present.) Now, the SentencePiece tokenization of the string $w=ababababa$ is $aba\tok b\tok aba\tok b\tok a$ (see~\cite{berg2023}) and the HuggingFace tokenization (refer to Algorithm~\ref{alg:hf}) is $ab\tok ab\tok ab\tok ab\tok a$, which the token DFA $A$ obviously accepts. Strictly speaking, for non-proper dictionaries, HuggingFace is not forced to select rules in decreasing order of priority (where for proper, it will by default), but only select the highest priority applicable rule, and keep on applying it left to right (see~\cite{berg2023}). Thus, with this formulation, the HuggingFace tokenization will be $ab\tok ab\tok ab\tok aba$, and thus, token DFA $A$ accepts neither HuggingFace nor SentencePiece tokenization, since $\delta(\cdot,aba)$ is undefined for any given state in $A$.
\begin{figure}
    \centering
    \begin{tikzpicture}[autbase]
\node[state, initial, accepting] (q0) {$q_0$};
\node[state, accepting, right of=q0] (q1) {$q_1$};
\draw (q0) edge[loop above] node{$ab$, $b$} (q0)
(q0) edge[bend left, above] node{$a$} (q1)
(q1) edge[loop above, above] node {$a$} (q1)
(q1) edge[bend left, below] node{$ab$} (q0);
    \end{tikzpicture}
    \caption{Automaton $A$ generated for dictionary  $D=[ab\tok a, a\tok b]$ by applying the merge $ab\tok a$ followed by the merge $a\tok b$ to the base token DFA.}
    \label{fig:nonproperdict}
\end{figure}
\end{example}

Finally, we are ready to prove that if we use Algorithm~\ref{alg:merge} iteratively to produce each time the token DFA $A_i$, using $A_{i-1}$ and the rule $u_i\tok v_i$ from the proper dictionary $D=[u_1\tok v_1, \ldots, u_n\tok v_n]$ as input (with $A_0$ being the base token DFA of a string DFA $A$), then the language of $A_i$ is the tokenization of $\mathcal{L}(A)$
using $D_i=[u_1\tok v_1,\ldots u_i\tok v_i]$.

\begin{theorem}
    \label{thm:all-of-it}%
    Let $A$ be a string DFA over base alphabet $\Sigma$. Let $D$ be the proper dictionary $D=[u_1\tok v_1, \ldots, u_n\tok v_n]$ over $\Sigma$. Furthermore, let $A_0, \ldots, A_n$ be the sequence of token DFAs where $A_0$ is the base token DFA for $A$ and for each $i>0$, $A_i$ is produced by applying the merge $u_i\tok v_i$ to $A_{i-1}$. Then, $\mathcal{L}(A_i)=\mathbb{T}^{D_i}(\mathcal{L}(A))$ for all $i\in\{0,\ldots,n\}$.
\end{theorem}

\begin{proof}
    We proceed by induction. Let $L=\mathcal{L}(A)$. It is given that $\mathcal{L}(A_0)=\mathbb{T}^\varnothing(L)$. Assume that $\mathcal{L}(A_i)=\mathbb{T}^{D_i}(L)$. Observe that since $A$ is a string DFA, $A_0$ is trivially context-invariant. Hence, each of the token DFAs in the sequence $A_1,\ldots, A_n$ must also be context-invariant by Lemma~\ref{lem:prefixpreserved}. It further follows from Lemma~\ref{lem:pairwisecompalt} that for every $\varphi\in \mathcal{L}(A_i)$  there is some $\varphi'\in \mathcal{L}(A_{i+1})$ such that $\pi(\varphi)=\pi(\varphi')$. Moreover, since $\mathbb{T}^{D_i}(\pi(\varphi))=\varphi$, then $\mathbb{T}^{D_{i+1}}(\pi(\varphi))=\varphi'$ according to Lemma~\ref{lem:correct-tok}. So, for all $\varphi\in\mathbb{T}^{D_i}(L)$, there exists some $\varphi'\in \mathcal{L}(A_{i+1})$ such that  $\varphi'\in\mathbb{T}^{D_{i+1}}(L)$. It only remains to show that all $\psi\in \mathcal{L}(A_{i+1})$ are in $\mathbb{T}^{D_{i+1}}(L)$.  If $\psi\in \mathcal{L}(A_{i+1})$ then by Lemma~\ref{lem:pairwisecompalt} $\psi'\in \mathcal{L}(A_i)$ so that  $\pi(\psi)=\pi(\psi')$, meaning $\psi'\in\mathbb{T}^{D_i}(L)$. We derive $\psi\in\mathbb{T}^{D_{i+1}}(L)$ from Lemma~\ref{lem:correct-tok}, hence $\mathcal{L}(A_{i+1})=\mathbb{T}^{D_{i+1}}(L)$. This completes the induction step.
\end{proof}

The case where $A$ is a universal string DFA is of particular interest, as the produced token DFA then accepts the language of all correct tokenizations.

\section{Complexity Bounds and Properties Preserved}\label{sec:complexitybounds}
A cursory reading of Algorithm~\ref{alg:merge} might suggest that \emph{each} application of the algorithm could double the size of the token DFA. This indeed holds for DFAs in which all states belong to $S_2$ (using the notation from Algorithm~\ref{alg:merge}), such as the one-state universal automaton: Algorithm~\ref{alg:merge} creates a fresh copy of each state in $Q$. Thus, when applying Algorithm~\ref{alg:merge} repeatedly, as in Theorem~\ref{thm:all-of-it}, the number of states could increase by a factor $2^{|D|}$. However, this turns out to be too pessimistic, as the actual growth is bounded by an additive factor, which we will call \emph{the degree of locality} of the DFA. Importantly, Algorithm~\ref{alg:merge} does not increase the degree of locality nor, more generally, the degree of $k$-locality of a token DFA.
\begin{definition}
    \label{defn:deg-k-loc}%
   A token DFA $A=(Q,\Gamma,q,\delta,F)$ has \emph{degree of $k$-locality $l$,} for $k,l\in \nat$, if for all $\tau\in \Gamma^*$, with $|\tau|=k$, we have $|\{ q' \in Q \mid q\in Q, q \xto{\tau} q'\}|\le l$.
\end{definition}

That is, if a (token) DFA has degree of $k$-locality $l$, then for any tokenization $\tau$  with $|\tau|\geq k$, there are at most $l$ possible states that a run reading $\tau$ can reach, regardless of the starting state of that run.

Let $\degloc(A,k)$ denote the \emph{smallest} $l\in \nat$ such that $A$ has a degree of $k$-locality $l$.

\begin{observation}
    \label{obs:degloc-monotonic-decreasing}
     For any DFA $A$, we have $\degloc(A,i)\ge \degloc(A,j)$ for all $i<j$. That is, $A$ cannot contain \emph{more} runs for a longer tokenization than for any of its subtokenizations.
\end{observation}

Definition~\ref{defn:deg-k-loc} generalizes the notion of a $k$-local automaton, closely related to several other concepts. First, recall the definition of a $k$-local DFA from~\cite{locally-testable-languages} (which, in turn, draws on~\cite{beal-local}).
\begin{definition}
    A token DFA $A=(Q,\Gamma,q,\delta,F)$ is \emph{$k$-local}, for $k\in \nat$, if for all tokenizations $\tau\in \Gamma^*$ with $|\tau|=k$, we  have $|\{ q' \in Q \mid q\in Q, q \xto{\tau} q'\}|\le 1$.
\end{definition}
A $k$-local automaton is thus an automaton with \emph{degree of $k$-locality 1.} Of particular interest are the \emph{local languages}~\cite{local-berry-sethi} the languages which are recognized
    by 1-local automata (also called local automata).\footnote{
    The reader may also be familiar with the $k$-testable languages, a hierarchy of language classes characterized by a membership problem decidable by inspecting only the set of $k$-length substrings of the input. The hierarchy of languages accepted by $k$-local automata is closely related to the hierarchy of \emph{strictly $k$-testable languages}, in particular if a language has a $k$-local automaton it is strictly $k'$-testable for some $k'$, but the exact relationship between $k$ and $k'$ depends on the definition used~\cite{locally-testable-languages}.}

Before proceeding, we introduce the notation required in the remainder of this section. For a token DFA $A=(Q,\Gamma,q_0,\delta,F)$, let $E_{\tau}(A)$ denote the set of states reached last in all runs reading $\tau$. Thus, we have {$\degloc(A,k)=\max_{|\tau|=k}|E_{\tau}(A)|$.} Further, for a tokenization $\phi$, $E_{\tau,\phi}(A)$ is the subset of states in $E_{\tau}(A)$ from which a run reading $\phi$ starts, while $E_{\tau}^{\phi}(A)\subseteq E_{\tau}(A)$ is its complement. Observe the distinction between $E_{uv}(A)$ and $E_{u,v}(A)$. The former denotes the set of states having an incoming transition on a token $uv$, whereas the latter denotes the set of states with an incoming transition on $u$ \emph{and} an outgoing transition on $v$.

Henceforth, when using Algorithm~\ref{alg:merge} in results and proofs, we take as the input DFA $A=(Q,\Gamma,q_0,\delta,F)$ and the rule $u \tok v$ (unless explicitly specified otherwise). Similar to Lemma~\ref{lem:pairwisecompalt}, the results in the remainder of this section trivially hold when $A$ already contains a transition on $uv$ due to its presupposed context-invariance (in which case $A=A'$). Therefore, we assume $uv\notin\Gamma$ and let $A'=(Q',\Gamma\cup \{uv\},q_0,\delta',F')$. We do not trim useless states from $A'$.

\begin{lemma}\label{lem:degkloc-nonincreasing}
Assume that Algorithm~\ref{alg:merge} is applied to the context-invariant DFA $A$ to produce $A'$ as output. Then $\degloc(A',k)\le\degloc(A,k)$.
\end{lemma}

\begin{proof}
We reuse the argument in the proof of Lemma~\ref{lem:pairwisecompalt}, relating the runs in $A$ to those in $A'$.

Suppose $\tau'$ is a tokenization read by a run in $A'$ of length $k'$. Let $\tau$ be the tokenization such that applying the rule $u\tok v$ to $\tau$ results in $\tau'$. If $\tau$ contains $u\tok v$, then $k'\leq k=|\tau|$, otherwise, $\tau=\tau'$. Note, that since $k'\leq k$, $\degloc(A',k)\le\degloc(A',k')$.

Let $\gamma$ be defined as in the proof of Lemma~\ref{lem:pairwisecompalt}, i.e., $\gamma$ recovers the state $q\in Q$ in $A$ from which the `fresh' state $q'\in Q'\setminus Q$ in $A'$ is created, and acts as the identity for all other states in $Q\subseteq Q'$.

The proof of Lemma~\ref{lem:pairwisecompalt} shows that if $p_0,\ldots,p_k$ is the sequence of states appearing in a run $p_0 \xto{\tau'} p_k$ in $A'$, then $\gamma(p_0),\ldots,\gamma(p_k)$ is a subsequence of states in a run $\gamma(p_0) \xto{\tau} \gamma(p_k)$ in $A$. Although $\gamma$ is not generally injective, it becomes injective when restricted to the domain to $Q$ or $Q'\setminus Q$, respectively.
Observe that if the last token in $\tau'$ is $u$, all runs reading $\tau'$ in $A'$ end in states in $Q'\setminus Q$. Otherwise, they end in states in $Q$. Thus, for any choice of $\tau'$, $\gamma$ is an injective mapping from $E_{\tau'}(A')$ to $E_{\tau}(A)$, ensuring $|E_{\tau'}(A')|\le |E_{\tau}(A)|$ and thus $\degloc(A',k')\leq \degloc(A,k)$.

 Combining that $\degloc(A',k)\leq\degloc(A',k')$ and $\degloc(A',k')\leq \degloc(A,k)$, we establish that $\degloc(A',k)\leq \degloc(A,k)$.
\end{proof}

In particular, this means that the algorithm preserves $k$-locality.
\begin{corollary}
    Algorithm~\ref{alg:merge} applied the $k$-local context-invariant token DFA $A$ results in a token DFA which is $k$-local.
\end{corollary}
\begin{proof}
    Since $A$ is $k$-local we have $\degloc(A,k)\le 1$, and thus $\degloc(A',k)\le 1$, i.e.\ $A'$ is $k$-local.
\end{proof}

We next relate the number of states added by a merge as in Algorithm~\ref{alg:merge}, to the number of states in the original DFA $A$.

\begin{lemma}
    \label{lem:state-complexity-old}%
    Assume Algorithm~\ref{alg:merge} is applied to the context-invariant token DFA $A$ and rule $u \tok v$ to produce $A'$. Then:
    \begin{enumerate}
        \item $|Q'|= |Q|+|Q'\setminus Q|=|Q|+|E_{u,v}(A)|\le |Q| + |E_{u}(A)|\le |Q|+\degloc(A,1)$;
        \item $|E_u(A)|=|E_{u,v}(A)|+|E_u^v(A)|=|Q'\setminus Q|+|E_u^v(A)|=|E_u(A')|;$
        \item $|E_{uv}(A')|\le |E_v(A)|$; and,
        \item $|E_{\beta}(A')|=|E_{\beta}(A)|$, for all $\beta\in \Gamma$.
    \end{enumerate}
\end{lemma}

\begin{proof}
    Observe that the set $S_2$ in Algorithm~\ref{alg:merge} is such that $|S_2|=|E_{u,v}(A)|\leq|E_u(A)|$, since $S_2$ consists exactly of those states with an incoming transition on $u$ and an outgoing transition on $v$. Since $|Q'|$ is the sum of $|Q|$ and $|S_2|$ (as a new state is created for every state in $S_2$), the equalities and inequalities in (1) must hold.

    For (2), by definition, $E_u(A)=E_{u,v}(A)\cup E_u^v(A)$ with $E_{u,v}(A)\cap E_u^v(A)=\varnothing$, which directly implies $|E_u(A)|=|E_{u,v}(A)|+|E_u^v(A)|$. Again, using the fact that $E_{u,v}(A)=S_2$ and that $|S_2|=|Q'\setminus Q|$, we obtain $|Q'\setminus Q|=|E_{u,v}(A)|$. Since Algorithm~\ref{alg:merge} ensures that $E_{u,v}(A')=\varnothing$, it follows that $E_u(A')=E_u^v(A')=(Q'\setminus Q)\cup E_u^v(A)$, and leading to the equality $|(Q'\setminus Q)|+|E_u^v(A)|=|E_u(A')|$.

    In the case of (3), $|E_v(A)|$ serves as an upper bound on the number of states that gain an incoming transition on $uv$ in $A'$. Algorithm~\ref{alg:merge} only adds transitions labelled $uv$ when $A$ transitions to a state $s_3$ (lines~\ref{alg:merge:add-uv} and~\ref{alg:merge:copy-to-s2p}), which by definition, must have an incoming transition on $v$.

    The equality in (4) is a matter of inspecting each step of Algorithm~\ref{alg:merge}. For $\beta\notin \{u,v\}$, only step~\ref{alg:merge:copy-to-s2p} of Algorithm~\ref{alg:merge} adds such transitions, where if $\delta(s_2,\beta)=s'$ we add a transition $\delta'(\fresh(s_2),\beta)=s'$. However, since $s'$ is by construction already a target for a transition labelled $\beta$, this does not increase $|E_{\beta}(A')|$. 
    For $\beta=u$, the result follows from (2).
    Finally, trivially $|E_v(A')|=|E_v(A)|$ as Algorithm~\ref{alg:merge} never adds or removes a transition on $v$.
\end{proof}

As a result the algorithm preserves the degree of 1-locality.
\begin{corollary}\label{cor:preservation_locality}
If the token DFA $A'$ is obtained by applying Algorithm~\ref{alg:merge} to the context-invariant token DFA $A$ we have $\degloc(A',1)=\degloc(A,1)$.
\end{corollary}
\begin{proof}
 When we set $k=1$ in the inequality $\degloc(A',k)\le\degloc(A,k)$, we obtain $\degloc(A',1)\le\degloc(A,1)$. But it follows from parts (3) and (4) of the previous lemma that this inequality is in fact an equality.
\end{proof}

\begin{theorem}
    \label{thm:tok-complexity}%
    Let $A$ be a context-invariant token DFA over $\Gamma$ with $n$ states, and let $D$ be a proper dictionary. Then, iteratively merging all rules from $D$ into $A$ (as in Theorem~\ref{thm:all-of-it}) produces a DFA with at most $n+|D|\degloc(A,1)$ states. There exist $A$ and $D$ (with $D$ of any specified finite cardinality) achieving this bound.
\end{theorem}
\begin{proof}
    We use the inequality $|Q'|\le |Q|+\degloc(A,1)$, from Lemma~\ref{lem:state-complexity-old} (1), to obtain the result. 
    
    If $|D|=0$, then $A=A'$, and the result trivially holds. Otherwise, merging the highest priority rule of $D$ into $A$, produces $A'$. By case 1 of Lemma~\ref{lem:state-complexity-old}, $A'$ has at most $k=\degloc(A,1)$ more states than $A$. Further, by Corollary~\ref{cor:preservation_locality}, $\degloc(A,1)=\degloc(A',1)$.
    Repeating this process with $A'$ and the next highest priority rule in $D$, we iterate for $|D|$ steps, obtaining the stated bound.

    Let $A$ be the single-state DFA for $a^*$, as shown on the left of Figure~\ref{fig:bound-reach}. For each $k=1,2,\ldots,$ construct the $k$-rule dictionary $D_k$ as follows:
    \[
    D_k = [a \tok a, aa \tok aa, aaaa \tok aaaa, \ldots, a^{2^{(k-1)}} \tok a^{2^{(k-1)}}].
    \]
    We argue that for each $k$, applying Theorem~\ref{thm:all-of-it} to $A$ with $D_k$ yields an automaton with exactly $k+1$ states. This precisely matches the stated bound since $A$ has $n=1$ and $\degloc(A,1)=1$. 
    
    The resulting automata for $D_1$, $D_2$ and $D_3$ are shown in Figure~\ref{fig:bound-reach}.
    \begin{figure}
        \centering
        \begin{tikzpicture}[autbase,node distance=1.25cm]
            \node[state, initial, accepting] (q0) {$q_0$};
            \draw (q0) edge[loop below] node {$a$} (q0);

            \node[state,initial,accepting] (q0) at (1.75,0) {$q_0$};
            \draw (q0) edge[loop below] node {$aa$} (q0);            
            \node[state,accepting,right of=q0] (q1) {$q_1$};
            \draw (q0) edge[above] node {$a$} (q1);

            \node[state,initial,accepting] (q0) at (5,0) {$q_0$};
            \draw (q0) edge[loop below] node {$aaaa$} (q0);
            \node[state,accepting,right of=q0] (q1) {$q_1$};
            \draw (q0) edge[below] node {$aa$} (q1);
            \node[state,accepting,above of=q1] (q2) {$q_2$};
            \draw (q0) edge[above left] node[xshift=3pt] {$a$} (q2);
            \draw (q1) edge[right] node {$a$} (q2);

            \node[state,initial,accepting] (q0) at (8.5,0) {$q_0$};
            \draw (q0) edge[loop below] node {$aaaaaaaa$} (q0);
            \node[state,accepting,right of=q0,xshift=10pt] (q1) {$q_1$};
            \draw (q0) edge[below] node {$aaaa$} (q1);
            \node[state,accepting,above of=q0] (q2) {$q_2$};
            \draw (q0) edge[left] node {$aa$} (q2);
            \draw (q1) edge[below left] node[pos=0.7,yshift=2pt,xshift=2pt] {$aa$} (q2);
            \node[state,accepting,above of=q1] (q3) {$q_3$};
            \draw (q0) edge[below right] node[xshift=-3pt,pos=0.75,yshift=2pt] {$a$} (q3);
            \draw (q1) edge[right] node {$a$} (q3);
            \draw (q2) edge[above] node {$a$} (q3);
        \end{tikzpicture}
        \caption{The first four of the automata which are used to demonstrate that the bound of Theorem~\ref{thm:tok-complexity} can be reached. From left to right the automaton accepting $a^*$, then the three automata resulting from successively applying the merges $a\tok a$, then $aa\tok aa$, and finally $aaaa \tok aaaa$.}
        \label{fig:bound-reach}
    \end{figure}
    To formalize, let $A_i$ be the automaton constructed for $\mathbb{T}^{D_i}(A)$. As the base case, observe that $A_1$ (second from the left in Figure~\ref{fig:bound-reach}) has two states, with a single transition on $aa$ forming a loop on the initial state. Now, assume that $A_i$ has $i+1$ states and that the only transition on $a^{2^i}$ is a loop on its initial state, denoted $q_0$. By Theorem~\ref{thm:all-of-it}, $A_{i+1}$ is obtained from $A_i$ by applying Algorithm~\ref{alg:merge} using the rule $a^{2^i}\tok a^{2^i}$.
    
    Since the loop on $q_0$ is the only occurrence of $a^{2^i}$ by assumption, $S=\{(q_0,q_0,q_0)\}$ in Algorithm~\ref{alg:merge}. Therefore, the algorithm introduces a single new state $\fresh(q_0)$, which is necessarily useful. This results in a total of $i+2$ states in $A_{i+1}$. Moreover, the only transition on $a^{2^{i+1}}$ is again a loop on $q_0$ in $A_{i+1}$, thereby verifying the inductive hypothesis.
\end{proof}

\begin{theorem}
    Algorithm~\ref{alg:merge} can be applied to a suitably encoded context-invariant token DFA $A=(Q,\Gamma,q_0,\delta,F)$ in time $\mathcal{O}(|Q||\Gamma|)$. Applying Algorithm~\ref{alg:merge} iteratively as in Theorem~\ref{thm:all-of-it} with a proper dictionary $D$ can be done in time $\mathcal{O}(|Q||\Gamma'||D|^2)$, where $\Gamma'=\Sigma\cup\{ uv\mid u\tok v\in D\}$.
\end{theorem}
\begin{proof}
    We assume that looking up $\delta(q,\alpha)$  and adding or removing transitions can be done in time $\mathcal{O}(1)$, and that a state can be added in time $\mathcal{O}(|\Gamma'|)$. The time complexity of Algorithm~\ref{alg:merge} is dominated by the loop over $S_2$ from lines 11 to 17, running $\mathcal{O}(|Q|)$ times, with each iteration running in time $\mathcal{O}(|\Gamma'|)$.  
    Furthermore, we can construct $S$ in time $\mathcal{O}(|Q|)$ by for each $q\in Q$ looking up $\delta(q,u)$ and $\delta(\delta(q,u),v)$. If both are defined, these states form a triple in $S$.
 Iterating the procedure, Theorem~\ref{thm:tok-complexity} dictates that no automaton in the process will have more than $|Q|+l|D|$ states, where $l=\degloc(A,1)$, with the cardinality of the token alphabet increasing from $|\Gamma|$ to $|\Gamma'|$. Putting it together, we apply Algorithm~\ref{alg:merge} $|D|$ times, each with a running time in $\mathcal{O}((|Q|+l|D|)|\Gamma'|)$, giving a time complexity of $\mathcal{O}(|Q||\Gamma'||D|+l|\Gamma'||D|^2)$, which is $\mathcal{O}(|Q||\Gamma'||D|^2)$, since $l\le |Q|$.
\end{proof}

Beyond the state complexity, the proof that the algorithm does not increase the degree of $k$-locality of automata also has other useful implications.
\begin{observation}
    $\mathbb{T}^D(\Sigma^*)$ is a local language for all $D$, since $\Sigma^*$ is a local language ($\degloc(A,1)=1$ for the single-state universal DFA $A$) and Corollary~\ref{cor:preservation_locality} dictates that this degree is preserved. As such, correct tokenization can be recognized by inspecting only adjacent pairs of tokens (plus the knowledge of which token is first). This is potentially useful in streaming applications, since a network device can check the correctness of a long tokenization using only a sliding window of width 2.

    Generally, for an automaton with $k$-locality 1, it is possible to verify a stream of tokens by checking that each $(k+1)$-length substring has \emph{some} run in the automaton.
\end{observation}

\section{Building a Transducer}
\label{sec:transducer}%
The procedure in Section \ref{sec:constructionprocedure} yields a token DFA $A$ for any given proper dictionary $D$ and regular language $L$ such that $\mathcal{L}(A)=\mathbb{T}^D(L)$. However, in most cases, the goal is not merely to recognize valid tokenizations but to perform the tokenization itself. To this end, the token DFA $A$ constructed in Section~\ref{sec:constructionprocedure} can be used to generate a transducer $T$ that explicitly maps input strings to their tokenized form, representing the relevant subset of the relation $\mathbb{T}^D$.

First, we construct a nondeterministic transducer $T$, which, when used directly for tokenization, turns out to be computationally expensive. Fortunately, $T$ can be transformed into a subsequential (i.e.\ input-deterministic) transducer that runs in linear time. 

We begin by recalling the definition of a (restricted form of a) transducer.
\begin{definition}
    A \emph{transducer} over an input alphabet $\Sigma$ and output alphabet $\Gamma$,  is a tuple $T=(Q,A,q_0,\delta,F)$, where each component in the tuple (mostly) serves the same purpose as for DFA in Definition~\ref{defn:dfa}, only now:
    \begin{itemize}
    \item the alphabet consists of pairs $A=\Sigma \times (\Gamma \cup \{\eps\})$ over $\Sigma$ (the \emph{input alphabet}) and $\Gamma$ (the \emph{output alphabet}); we denote $(u,v) \in A$ as $u|v$ in rules and runs;
    \item $\delta \subseteq Q \times A \times Q$ is a not necessarily functional relation; and,
    \item $F : Q \to \Gamma^*$ is a possibly partial function; we say that the states in the domain of $F$ are the accepting states.
    \end{itemize}
    We use the same notation for transitions and runs as in Definition~\ref{defn:dfa}. Specifically, for $(q,(a,b),q') \in \delta$, we write $q \xto{a|b} q'$, meaning that from $q$, reading $a$ produces output $b$. 
    
    A run $q_1 \xto{u_1|v_1} \cdots \xto{u_n|v_n} q_{n+1}$ in $T$, reads the input $u_1 \cdots u_n$ and outputs $v_1 \cdots v_n$. (As usual, $\eps$ is the identity, so $u\cdot \eps=\eps \cdot u$ for all $u$.) When the intermediate states are not of interest, the run may be written more compactly as $q_1 \xto{u_1\cdots u_n|v_1\cdots v_n} q_{n+1}$. 
    
    An accepting run $q_0 \xto{u|v} q_n$ must start in $q_0$ and end in a state where $F$ is defined; such a run reads $u\in \Sigma^*$ and outputs $v\cdot F(q_n) \in \Gamma^*$.

    The \emph{relation recognized by $T$} is the set of all $(u,v)\in \Sigma^* \times \Gamma^*$ such that there is an accepting run in $T$ that reads input $u$ and outputs $v$. We denote this set by $\mathcal{L}(T)$.

    A transducer is \emph{subsequential} (or equivalently, input-deterministic) if  for every $(p,u)\in Q\times\Sigma$, we have $|\{(v,q) \mid (p,(u,v),q) \in \delta\}|\le 1$.  The transducer is called \emph{sequential} if it additionally satisfies $F(q)=\eps$ for all $q$ on which $F$ is defined.
\end{definition}
Observe that this definition of a transducer is somewhat restricted (notably, it is necessarily real-time, see~\cite{beal2002} for more details), tailored to the needs of our disposition.

\begin{observation}
\label{obs:build-trans}%
For a token DFA $A=(Q,\Gamma,q_0,\delta,F)$, with base alphabet $\Sigma$, we construct a transducer $T=(Q',\Sigma\times (\Gamma\cup\{\eps\}),q_0,\delta',F')$ that recognizes $\{(\pi(\varphi),\varphi) \mid \varphi \in \lang(A)\}$ as follows. 

Initially, let $Q'=Q$. For each transition $p\xto{u} q$ in $A$, where $u=\alpha_1\cdots \alpha_k\in\Gamma$ and $\alpha_i\in\Sigma$, replace it with a $k$-step run in $T$ that produces the token $u$ upon reading the string $u$ from state $p$. Specifically, add $k-1$ intermediate states $s_1, \ldots, s_{k-1}$ to $Q'$, with $s_0$ denoting $p$. Next, define $\delta'$ to be the union of $\{(s_i,(\alpha_{i+1},\eps),s_{i+1}) \mid 0\le i < k-1\}$ and $\{(s_{k-1},(\alpha_k,u),q)\}$.

Finally, we let $F'(q)=\eps$ for all $q\in F$, and leave the final state function undefined otherwise.
\end{observation}

\begin{lemma}
    \label{lem:func}%
    The transducer $T$ obtained by applying Observation~\ref{obs:build-trans} to $A$ has $\lang(T)=\{(\pi(\varphi), \varphi) \mid \varphi \in \lang(A)\}$, and $T$ is necessarily functional if $A$ is context-invariant.
\end{lemma}
\begin{proof}
    This follows directly from the construction. For an accepting run $q_0 \xto{u|v} q_n$ in $T$, we demonstrate that $\pi(v)=u$ and that there is an accepting run $q_0 \xto{v} q_n$ in $A$. 
    
    Observe that both $q_0$ and $q_n$ must be states in $A$, since none of the intermediate states introduced during the construction of $T$ are initial or accepting. Now, consider any subrun $p \xto{u'|v'} q$ such that $p$ and $q$ are in $A$, but none of its intermediate states. By construction, we know that $|v'|=1$ \emph{and} that $p \xto{v'} q$ is a valid transition in $A$. We can inductively find the accepting run in $A$ by replacing each subrun with the corresponding transition in $A$. The reverse direction proceeds similarly.

    The transducer $T$ is functional for all context-invariant $A$. To see this, take any pair $(x,y)$ and $(x,y') \in \lang(T)$. We already established that $\pi(y)=\pi(y')$, and that $A$ must have runs on both $y$ and $y'$. Since $A$ is context-invariant, it follows that $y=y'$, ensuring that $T$ functional.
\end{proof}

Interestingly, while the transducer $T$ obtained by applying Observation~\ref{obs:build-trans} to token DFA $A$ is functional, it is not necessarily subsequential.

\begin{example}\label{ex:tr_a_tok_a}
    Consider the token DFA $A_1$ illustrated in Fig.~\ref{fig:constructionprocedureex}. Applying Observation~\ref{obs:build-trans} to it produces a transducer where the initial state has two outgoing edges on the input symbol $a$, see Fig.~\ref{fig:tokenizertransducer}(a).

    \begin{figure}[htb]
        \centering
        \begin{tikzpicture}[autbase,node distance=1.25cm]
            \node[state, initial,accepting] (q0) {$q_0$};
            \node[state, above of=q0] (s2)  at ($(q0)+(0cm,0.8cm)$) {$s_1$};
            \node[state, right of=q0,accepting] (q1) at ($(q0)+(0.8cm,0cm)$) {$q_1$};
            
            \draw (q0) edge[bend left, left,pos=0.7] node{$a|\varepsilon$} (s2);
            \draw (s2) edge[bend left, right,pos=0.3] node{$a|aa$} (q0);
            \draw (q0) edge[loop below] node{$b|b$} (q0);
            \draw (q0) edge[above,bend left,pos=0.6] node {$a|a$} (q1);
            \draw (q1) edge[below,bend left,pos=0.4] node {$b|b$} (q0);

            \node[xshift=-1cm, yshift=-1.2cm] (label1) at (q0) {(a)};
            \node[xshift=4cm, yshift=-1.2cm] (label1) at (q0) {(b)};
            \node[state, initial,accepting] (q0) at ($(q0)+(5cm,0cm)$) {$q_0$};
            \node[state, right of=q0,accepting] (q1) at ($(q0)+(1.2cm,0cm)$) {$q_1$};
            \node[state,draw=none] (q2) [right of=q1] {};
            
            \draw (q0) edge[loop below] node{$b|b$} (q0);
            \draw (q0) edge[above,bend left,pos=0.5] node {$a|\eps$} (q1);
            \draw (q1) edge[below,bend left,pos=0.5] node {$a|aa$, $b|b$} (q0);
            \draw (q1) edge[above,pos=0.5] node {$a$} (q2);

        \end{tikzpicture}
        \caption[Tokenization transducer for token DFA $A_1$]{Two tokenizing transducers describing $\mathbb{T}^D$ for $D=[a \tok a]$. The corresponding universal token DFA over $\Sigma=\{a,b\}$ is $A_1$ in Figure~\ref{fig:constructionprocedureex}. In (a) the transducer obtained by applying Observation~\ref{obs:build-trans} directly to $A_1$, is given. The underlying input automaton of this transducer is not deterministic, since reading $a$ at $q_0$ can take us both to the new intermediate state $s_1$ and to $q_1$. In (b) an equivalent transducer with a deterministic underlying input automaton is given, i.e.\ it is subsequential. It is, however, not sequential, note that $q_1$ is accepting; the extra arrow denotes that the accepting function has $F(q_1)=a$ (all other states $q$ drawn accepting have $F(q)=\eps$).}
        \label{fig:tokenizertransducer}
    \end{figure}

\end{example}

The remainder of this section aims to show that although Observation~\ref{obs:build-trans} does not yield a subsequential transducer, such transducer must exist. Specifically, we leverage the fact that a (partial) function with bounded variation can be realized by a subsequential transducer.

\begin{definition}[Bounded Variation] A partial function $f\colon A^*\to B^*$ has bounded variation if and only if \[\forall k\geq 0, \exists K \geq 0, \forall u, v\in dom(f)\ d(u,v)\leq k \implies d(f(u),f(v))\leq K,\]
where the function $d(x,y)$ for elements $x,y$ over any alphabet is defined as
\[d(x,y)=|x|+|y|-2|x\land y|,\]
and $x\land y$ denotes the longest common prefix of $x$ and $y$.
\end{definition}

Given two strings $u,v \in \Sigma^*$ and a proper dictionary $D$, it is not necessarily the case that $\mathbb{T}^D(u\land v)=\mathbb{T}^D(u)\land \mathbb{T}^D(v)$.
\begin{example}
    Let $u=ababa$ and $v=ababb$, and $D=[a\tok b, ab\tok a]$ then $\mathbb{T}^D(u\land v)=\mathbb{T}^D(abab)=ab\tok ab$ but $\mathbb{T}^D(u)\land \mathbb{T}^D(v)=(ab\tok aba) \land (ab\tok ab\tok b)=ab$.
\end{example}

\begin{lemma}\label{lem:tokenhasboundedvariation}
    For any proper dictionary $D$, the function $\mathbb{T}^D$ has bounded variation.
\end{lemma}
\begin{proof}
    Let $t\coloneqq \max\{|uv|\mid u\tok v\in D\}$ be the maximum token length in dictionary $D$. 
    
    Fix $k\geq 0$, and consider two arbitrary strings $x,\ y$ such that $d(x,y)\leq k$. Suppose $w=x\land y$ is the longest common prefix, then 
    \[x=wx',\quad y=wy',\]
    for some suffix $x'$ and $y'$. Observe that $d(x,y)=|x'|+|y'|$, thus $|x'|+|y'|\leq k$. 
    
   By the dictionary look-ahead bound in~\cite{berg2023}, we have
    \[\mathbb{T}^D(x)=\mathbb{T}^D(wx')=\varphi\tok\mathbb{T}^D(w'x'),\]
    and 
    \[\mathbb{T}^D(y)=\mathbb{T}^D(wy')=\varphi\tok\mathbb{T}^D(w''y'),\]
    for some token sequence $\varphi$, satisfying $|\varphi|\geq |\mathbb{T}^D(w)|-|D|$ and $|w'|,\ |w''|\leq |D|\cdot t$. Note that if $\mathbb{T}^D(w)|\leq|D|$, then $\varphi$ may be empty.

    We bound variation of $\mathbb{T}^D$ as follows:
    \begin{align*}
        d(\mathbb{T}^D(x),\mathbb{T}^D(y))&=d(\mathbb{T}^D(w'x'),\mathbb{T}^D(w''y'))\\
        &\leq |\mathbb{T}^D(w'x')|+|\mathbb{T}^D(w''y')|\\
        &\leq |w'x'|+|w''y'|
        =|w'|+|x'|+|w''|+|y'|\\
        &\leq 2|D|\cdot t+k.
    \end{align*}
    Since $|D|$ and $t$ depend only on the fixed dictionary $D$, $K\coloneqq  2|D|\cdot t+k$ is a constant. Hence, $\mathbb{T}^D$ has bounded variation.
\end{proof}

\begin{example}
    While $\mathbb{T}^D$ has bounded variation, even small modifications to the tokenization semantics can cause it to violate this property. For example, consider changing ``minimizing $|\phi|$'' into ``maximizing $|\phi|$'' on line 5 in Algorithm~\ref{alg:hf} (i.e.\ rules now apply right-most instead of left-most). Now, take $D=[a \tok a]$, and consider the strings $a,aa,aaa,\ldots$\,. For $n\in \nat$, the correct tokenization of $a^n$ using these semantics is  $a\tok aa \tok aa \tok \cdots $ if $n$ is odd, and $aa\tok aa\tok \cdots$ if $n$ is even. Thus, intuitively, a transducer must determine if the input length is even or odd \emph{before} it can produce the first output token. However, a subsequential transducer cannot `guess' nor read the entire input before outputting any tokens. If the transducer processes the entire input first, it must keep track of the number of $aa$ tokens to be output at the end, which is impossible with a finite set of states. This argument can be made firm by a straightforward pumping argument.
\end{example}

\begin{example}
    Observe that $\mathbb{T}^D$ for a proper dictionary $D$ does not necessarily have a corresponding \emph{sequential} transducer. Consider the dictionary $D=[a\tok a]$ again (see Figure~\ref{fig:tokenizertransducer}), a sequential transducer can never output the token $a$, as it has no way to recognize the end of the string. Instead, a subsequential transducer for this relation would have $F(q)=a$ for some accepting state $q$ (see Example~\ref{ex:tr_a_tok_a}).
\end{example}

Lemma~\ref{lem:tokenhasboundedvariation}, it follows that we can obtain a subsequential transducer from the one produced in Observation~\ref{obs:build-trans}.
\begin{theorem}
    \label{thm:det-trans}%
    For all proper dictionaries $D$ there exists a subsequential string-to-string transducer $T'$ which describes the function $\mathbb{T}^D$.
\end{theorem}
\begin{proof}
    Follows since $\mathbb{T}^D$ has bounded variation~\cite{beal2002}. \cite{beal2002} also offers an algorithm that can be used to obtain $T'$ from a transducer $T$ constructed according to Observation~\ref{obs:build-trans}.
\end{proof}

Theorem~\ref{thm:det-trans} implies that we can improve the common priority queue implementation of tokenization, used e.g.\ in the SentencePiece tokenizer, which runs in time $\mathcal{O}(n\log n)$ for a string of length $n$.
\begin{corollary}
    Fixing a proper dictionary $D$, tokenization of a string of length $n$ with $D$ can be done in time $\mathcal{O}(n)$ by string-to-string transduction.
\end{corollary}

Fixing the dictionary $D$, results in a constant-sized transducer, though it may have a large state space, potentially making it overly permissive. In addition, pretokenization (splitting on white space) keeps in many practical cases the strings short. For more on tokenization using transduction, also consider the recent paper~\cite{cognetta2024tokenization}, which uses a different approach.

\section{Conclusions and Future Work}\label{sec:conclusion}

To summarize, we have given an algorithm for constructing finite automata which recognize correct BPE tokenizations, proven it correct, demonstrated it efficient, and shown how a subsequential string-to-tokenization transducer can be obtained from it. In immediate applications, the automaton enables efficient validation of tokenization correctness, efficient pattern matching on tokenized text, and equivalence checking of proper BPE dictionaries. The transducer can be used to perform tokenization, and can be composed with other languages in various ways. In general, the results make the many algorithms and closure properties of regular languages available in the tokenized setting.

As this is a practical and currently highly relevant area, there is a lot of potential future work to consider. Section~\ref{sec:complexitybounds} demonstrates that the state complexity of the automata constructed is quite modest; the base alphabet will in practice be Unicode code points, which means that each state may have a very large number of transitions. As such, the encoding of the automata should be done with some care, potentially investigating a link to symbolic automata. Similarly, the transducer produced by Observation~\ref{obs:build-trans} is of modest size, but the state complexity of producing the equivalent subsequential transducer should also be investigated with care.

\paragraph{Acknowledgments} This work was partially supported by the Wallenberg AI, Autonomous Systems and Software Program (WASP) funded by the Knut and Alice Wallenberg Foundation.

\bibliographystyle{splncs04}
\bibliography{references}  

\end{document}